\documentclass[9pt,journal]{IEEEtran}

\usepackage{pgfplots}
\pgfplotsset{compat=1.13} 

\usepackage{tikz}
\usepgfplotslibrary{external} 
\usepackage{graphicx}
\usepackage{epstopdf}
\ifCLASSINFOpdf
\else
\fi

\usepackage{amsmath,amssymb} 
\usepackage{mathtools} 
\usepackage{cases}	
\usepackage{nicefrac}	
\usepackage[overload]{empheq}	

\usepackage{array}
\newcommand{\textkom}{\text{,}}
\newcommand{\txtfor}[1]{\text{$#1$}}

\newcommand{\tp}{\textnormal{\textsf{T}}}
\newcommand{\hp}{\textnormal{\textsf{H}}}
\newcommand{\set}[1]{\mathcal{#1} }
\renewcommand{\vec}[1]{\mbox{\boldmath{$#1$}}}
\newcommand{\mat}[1]{\mbox{\boldmath{$#1$}}}

\newcommand{\rank}[1]{\text{rank}\left(#1\right)}
\newcommand{\dop}{\text{ d}}

\newcommand{\cd}{(\cdot)}

\DeclareMathAlphabet\mathbb{U}{fplmbb}{m}{n}
\newcommand{\RZ}{\mathbb{R}}
\newcommand{\NZ}{\mathbb{N}}
\newcommand{\CZ}{\mathbb{C}}
\renewcommand{\Re}[1]{\mathfrak{Re}\{#1\}}

\def \oa {u}
\def \ua {y}

\def \zero {\vec{0}}                    
\def \I {\mat{I}}                       

\def \x {\vec{x}}                       
\def \xol {\overline{\x}}               
\def \xoli {\overline{x}}               
\def \u {\vec{u}}                       
\def \y {\vec{y}}                       
\def \yol {\overline{\y}}               
\def \xs {\vec{x}_\text{s}}             
\def \us {\vec{u}_\text{s}}             
\def \ys {\vec{y}_\text{s}}             
\def \zs {\vec{z}_\text{s}}             
\def \zh {\vec{z}_\text{H}}             
\def \zsi {\vec{z}_{\text{s},i}}             
\def \zhi {\vec{z}_{\text{H},i}}             
\newcommand{\dx}{\delta \x}				
\newcommand{\du}{\delta \u}				
\newcommand{\dJ}[1]{\delta^{#1}J}				
\newcommand{\phiv}{\vec{\phi}_\text{s}} 
\newcommand{\betav}{\vec{\beta}} 		

\def \A {\mat{A}}                       
\def \Aol {\overline{\A}}               
\def \lambdaol {\overline{\lambda}}  	
\def \B {\mat{B}}                       
\def \C {\mat{C}}                       
\def \Col {\overline{\C}}				
\def \D {\mat{D}}                 		
\def \Ed {\mat{E}_\text{d}}                 	
\def \Dd {\mat{D}_\text{d}}                 	

\def \mPi {\mat{\Pi}}                   
\def \mPic {\mat{\Pi}_{\phi}}           
\def \mGamma {\mat{\Gamma}}             
\def \mGammac {\mat{\Gamma}_\phi}             
\def \dPi {\delta\mPi}                   
\def \dGamma {\delta\mGamma}                   
\def \sdPi {\delta\vec{\pi}}

\def \mA {\mat{\set{A}}} 				
\def \mAol {\overline{\mat{\set{A}}}}
\def \DA {\mat{\Delta}} 			
\def \DAol {\overline{\mat{\Delta}}}
\def \vx {\mat{\set{X}}} 					
\def \tvx {\widetilde{\mat{\set{X}}}} 		
\def \vb {\mat{\set{B}}}  					

\def \Q {\mat{Q}}                       
\def \Qy {\mat{Q}}
\def \R {\mat{R}}                       

\def \RGoa {(\set{RE}^\oa)}				
\def \RGua {(\set{RE}^\ua)}				
\def \RG {(\set{RE})}					

\newcommand{\Ros}[1]{\vec{\set{R}}\left(#1\right)}		
\newcommand{\Ross}[1]{\vec{\set{R}}_\text{s}\left(#1\right)}		
\newcommand{\Rosc}[1]{\vec{\set{R}}_\text{c}\left(#1\right)}		
\newcommand{\Rosh}[1]{\vec{\set{R}}^\oa_\text{H}\left(#1\right)}		
\newcommand{\Roshu}[1]{\vec{\set{R}}^\ua_\text{H}\left(#1\right)}		

\def \s {\text{s}}				
\def \ct {\text{c}}				

\def \nol {\overline{n}}                

\def \K {\mat{K}}                       

\def \G {\mat{G}}            
\def \W {\mat{W}}
\def \M {\mat{M}}            

\def \nulls {\text{null}}
\def \lnulls {\text{leftnull}\left(\Ross{\lambdaol}\right)}
\def \lnullh {\text{leftnull}\left(\Rosh{\lambdaol}\right)}

\newcommand{ \ma}{\begin {bmatrix}}
\newcommand{\me }{ \end {bmatrix}}
\newcommand{ \mra}{\begin {pmatrix}}
\newcommand{\mre }{ \end {pmatrix}}

\newcommand{\thmref}[1]{Theorem~\ref{#1}}
\newcommand{\secref}[1]{Section~\ref{#1}}

\newcommand{\assref}[1]{Asmp.~\ref{#1}}
\newcommand{\remref}[1]{Remark~\ref{#1}}
\newcommand{\lemref}[1]{Lemma~\ref{#1}}

\newcommand{\appref}[1]{Appendix~\ref{#1}}

\newcommand{\lqtref}[1]{LQTP~\ref{#1}}
\newcommand{\orref}[1]{ORP~\ref{#1}}
\newcommand{\oorref}[1]{OORP~\ref{#1}}

\newtheorem{lemma}{Lemma}
\newtheorem{theorem}{Theorem}
\newtheorem{assumption}{Assumption}
\newtheorem{coroll}{Corollary}
\newtheorem{defn}{Definition}
\newtheorem{remark}{Remark}

\newtheorem{orprob}{Output Regulation Problem}
\newtheorem{oorprob}{Optimal ORP}
\newtheorem{lqtprob}{LQT Problem}

\def\BibTeX{{\rm B\kern-.05em{\sc i\kern-.025em b}\kern-.08em
    T\kern-.1667em\lower.7ex\hbox{E}\kern-.125emX}}
\begin{document}

\def\proof{\@ifnextchar[{\@proof}{\@proof[Proof]}}
\def\@proof[#1]{\noindent\hspace{2em}{\itshape #1: }}

\def \QED {\IEEEQED}	
\def\endproof{\hspace*{\fill}~\QED\par\endtrivlist\unskip}
	
\title{Optimal Output Regulation for Square, Over-Actuated and Under-Actuated Linear Systems}
\author{Sebastian Bernhard and J\"urgen Adamy
	\thanks{S. Bernhard and J. Adamy are with the Institute of Automatic Control and Mechatronics, Control Methods and Robotics Lab;
		Technische Universit\"at Darmstadt, Landgraf-Georg Str. 4, 64283 Darmstadt, Germany,
		{\tt \{bernhard, adamy\}@rmr.tu-darmstadt.de}}
}

\maketitle

\begin{abstract}
	This paper considers two different problems in trajectory tracking control for linear systems. First, if the control is not unique which is most input energy efficient. Second, if exact tracking is infeasible which control performs most accurately. These are typical challenges for over-actuated systems and for under-actuated systems, respectively. We formulate both goals as optimal output regulation problems. Then we contribute two new sets of regulator equations to output regulation theory that provide the desired solutions. A thorough study indicates solvability and uniqueness under weak assumptions. E.g., we can always determine the solution of the classical regulator equations that is most input energy efficient. This is of great value if there are infinitely many solutions. We derive our results by a linear quadratic tracking approach and establish a useful link to output regulation theory.
\end{abstract}

 \IEEEpeerreviewmaketitle
 
 \section{Introduction}
\label{sec_intro}
In many practical applications it is desired that the system output tracks a time-varying reference. In a quite general setting, output regulation theory states conditions when exact tracking with zero steady state error is possible and provides a simple way to calculate and to implement a control that achieves it, see \cite{Huang2004,Isidori2017,Saberi2000,Trentel2001}. 

However, such a solution is not unique for over-actuated processes as, e.g., hybrid electric vehicles \cite{Sciarretta2007}. Then it is of great interest, how this surplus of actuators can be used beneficially. Many publications consider this question, e.g., \cite{Krener1992,Moase2016, Serrani2012}. 
The converse problem, when exact tracking is impossible, also drew attention in recent years, see \cite{  Corona2018b,Davison2011,Garcia2011,Saberi2000}. The question arises: Which control yields the highest tracking accuracy? An appropriate answer is of great value with respect to under-actuated systems such as underwater vehicles \cite{Aguiar2007} or square systems affected by actuator failures \cite{Tsiotras2000b}.

In this paper, we establish a connection between both questions and an optimal tracking problem. In this way, we are able to give answers that are surprisingly concise and universal at the same time. Whereas results in the literature are often complex, our approach preserves the simplicity in the control structure and the design. Hence, it qualifies as a natural extension to output regulation theory for over-actuated systems as well as a natural bridge to under-actuated systems.

\subsection{Problem Formulation and Main Contribution}
\label{sec_contr}
We consider linear time-invariant systems of the form%
\begin{subequations}\label{eq_sys}\begin{align} 
	\dot{\x} &= \A\x+\B\u+\Ed \xol\textkom \label{eq_linsys}\\
	\y &= \C\x+\Dd \xol \label{eq_y}
	\end{align}\end{subequations}
where $\x(t)\in\RZ^n$ is the state with initial value $\x(0)=\x_0$, $\u(t)\in\RZ^m$ is the input and $\y(t)\in\RZ^p$ is the output for $t\in[0,\infty)$. We will consider a feedthrough $\D\u$ to the output later on. For now, suppose $\D=\zero$. The system is affected by state disturbances $\Ed \xol$ and output disturbances $\Dd \xol$. Together with the reference output trajectory $\yol(t)\in\RZ^{p}$, these are generated by an exosystem%
\begin{subequations}\label{eq_exo}\begin{align} 
	\dot{\xol} &= \Aol \xol\textkom \label{eq_exosys}\\ 
	\yol &= \Col \xol \label{eq_yol}
	\end{align}\end{subequations} 
with state $\xol(t)\in\RZ^{\nol}$ and initial value $\xol(0)=\xol_0$. The exosystem is usually block-diagonal in order to account for both tasks.

Then it is desired that the output \eqref{eq_y} tracks the reference \eqref{eq_yol} or, more precisely, we want to solve%
\begin{orprob}[ORP 1]\label{or}
	Find the matrices $\mPi\in\RZ^{n\times\nol}$, $\mGamma\in\RZ^{m\times\nol}$ and $\K\in\RZ^{m\times n}$ for which the control 
	\begin{equation}\label{eq_u}
		\u=-\K(\x-\mPi\xol)+\mGamma\xol
	\end{equation} 
	guarantees that the tracking error $\y-\yol$ is regulated such that
	\begin{equation}\label{eq_error}
		\lim_{t\to\infty}\y(t)-\yol(t)=\zero
	\end{equation} 
	holds for all $\x_0$ and $\xol_0$.
\end{orprob}

To solve this problem, the choice of $\mPi$ and $\mGamma$ is essential. Assume that system \eqref{eq_sys} is stabilizable and all eigenvalues of $\Aol$ lie in the closed right half-plane. Then, it is a well known result by \cite{Francis1977} that \orref{or} can be solved by a linear (dynamic) control law such as \eqref{eq_u} \textit{if and only if} a solution $(\mPi,\mGamma)$ to the \textit{classical regulator equations}
\begin{subequations}\label{eq_reg}\begin{align}[left = \RG\enspace\empheqlbrace]
	\mPi\Aol&=\A\mPi+\B\mGamma+\Ed \label{eq_reg1}\\
	\Col&=\C\mPi+\Dd \label{eq_reg2}
\end{align}\end{subequations}  
exists. Indeed, when we choose a feedback $-\K\x$ that stabilizes system \eqref{eq_sys}, then the state converges to its so-called \textit{stationary state} $\xs(t)\coloneqq\mPi\xol(t)$, i.e., $\lim_{t\to\infty}\x(t)-\mPi\xol(t)=\zero$. It is induced by the stationary control $\us(t)\coloneqq\mGamma\xol(t)$ in \eqref{eq_u} and the disturbance $\Ed\xol$ in \eqref{eq_linsys}. With respect to \eqref{eq_exosys}, we call $\big(\xs\cd,\us\cd\big)$ a \textit{stationary solution} of \eqref{eq_linsys}. As a result of \eqref{eq_reg2}, the tracking error vanishes asymptotically since the stationary output $\ys\coloneqq\C\xs+\Dd\xol$ satisfies $\ys-\yol\equiv\zero$. 

Considering \orref{or}, we are motivated by the following questions that may arise when we intend to solve the regulator equations:
\begin{enumerate}
	\item If their solution is not unique, which choice of $(\mPi,\mGamma)$ gives the control $\us$ that is most input energy efficient? \label{ques1}
	\item If $\RG$ cannot be solved, what control $\us$ should be chosen in order to keep the nonzero tracking error $\ys-\yol$ small?\label{ques2}%
\end{enumerate}
Question \ref{ques1}) arises in the context of \textit{over-actuated} system \eqref{eq_sys} for which $\rank{\B}>\rank{\C}$. Then we are of course interested in a solution of \orref{or} that uses the additional actuators beneficially. We have to face question \ref{ques2}) if \orref{or} is infeasible and nonzero tracking errors are unavoidable. This is typical for \textit{under-actuated} system \eqref{eq_sys} with $\rank{\B}<\rank{\C}$. 
Now, we reformulate the open questions by two optimal output regulation problems. These state reasonable goals for over- and under-actuated systems, respectively.
\begin{oorprob}[OORP 1]\label{oor}
	Find a pair $(\mPi^\oa,\mGamma^\oa)$ such that for \textit{every} $\xol_0$, $\big(\xs^*(\cdot),\us^*(\cdot)\big)=\big(\mPi^\oa\xol(\cdot),\mGamma^\oa\xol(\cdot)\big)$ is a stationary solution of \eqref{eq_linsys} that minimizes the power $P^u\big(\xs(\cdot),\us(\cdot)\big)\coloneqq\lim_{T\to\infty}\left(\frac{1}{T} J_T^u\big(\xs(\cdot),\us(\cdot)\big)\right)$ w.r.t. to the stationary input energy
	\begin{equation}\label{eq_Ju}
	J_T^u\big(\xs(\cdot),\us(\cdot)\big)=\tfrac{1}{2}\int_0^T \us(t){^\tp}\R\us(t) \dop t
	\end{equation}
	with $\R$ real symmetric positive definite
	under the constraint 
	\begin{equation}\label{eq_constr}
	\ys(t)=\Col\xs(t)+\Dd\xol(t)=\Col\xol(t)\textkom\;\; \forall t\in [0,\infty).
	\end{equation}
\end{oorprob}

\vspace{3mm}
We discuss the expansion of cost \eqref{eq_Ju} by a state penalty later on.
\begin{oorprob}[OORP 2]\label{uor}
	Find a pair $(\mPi^\ua,\mGamma^\ua)$ such that for \textit{every} $\xol_0$, $\big(\xs^*(\cdot),\us^*(\cdot)\big)=\big(\mPi^\ua\xol(\cdot),\mGamma^\ua\xol(\cdot)\big)$ is a stationary solution of \eqref{eq_linsys} that minimizes the power $P^y\big(\xs(\cdot),\us(\cdot)\big)\coloneqq\lim_{T\to\infty}\left(\frac{1}{T} J_T^y\big(\xs(\cdot),\us(\cdot)\big)\right)$ with stationary error energy
	\begin{equation}\label{eq_Jy}
	\hspace{-5mm}J_T^y\big(\xs(\cdot),\us(\cdot)\big)=\tfrac{1}{2}\int_0^T \big(\ys(t)-\yol(t)\big){^\tp}\Qy\big(\ys(t)-\yol(t)\big) \dop t\hspace{-3mm}
	\end{equation} 
	with $\Qy$ real symmetric positive definite.
\end{oorprob}

In fact, both problems are linear quadratic tracking (LQT) problems over infinite horizons $T\rightarrow\infty$ where the admissible solutions are restricted to be stationary solutions of \eqref{eq_linsys}. We consider an infinite horizon because our focus lies on the stationary behavior \eqref{eq_constr} rather than on the transient response. 
In more detail, \oorref{oor} comprises a state constraint \eqref{eq_constr}. Such kind of ``pure'' state constrained LQT problems are difficult to solve, e.g., see \cite{Bryson1975} and \cite{Grass2008}. 
As for cost \eqref{eq_Jy}, we observe that it is independent of the control $\us(\cdot)$ and, hence, \oorref{uor} is a singular LQT problem. Singular optimal control problems are more complicated in general and rather considered as regulator problems than as tracking problems, e.g., \cite{Anderson2007}, \cite{Bryson1975},  \cite{Grass2008}, \cite{Pratti2008}.%

To the best of our knowledge, neither of the two challenging \textit{optimal output regulation problems} has been rigorously solved in such a general setup yet. In particular, a solution method as simple as solving algebraic equations such as $\RG$ is \textit{not known} for either of the two problems in general as we will discuss in the next section. 

Beforehand, we highlight our \textbf{main contributions}:\\ 
Using a \textit{unifying} linear quadratic tracking approach for the first time, we give rigorous solutions to both optimal output regulation problems \ref{oor} and \ref{uor} under suitable assumptions. For each, we derive \textit{new regulator equations} that provide the desired solution in the most simple and natural way. We prove optimality and study conditions for solvability and uniqueness in detail. Putting our results in a wider context, we verify that the solution of each problem \orref{or}, \oorref{oor} and \ref{uor} is the limit of the solution of a special case of a classical linear quadratic tracking problem over an infinite horizon. Concluding, we contribute to output regulation theory by answering questions \ref{ques1}) and \ref{ques2}) thoroughly and by giving a natural extension to account for over-actuated and under-actuated systems in a general manner.

\subsection{State of the Art}
With respect to question \ref{ques1}), problems that are very similar to \oorref{oor} are studied in the literature when system \eqref{eq_sys} is over-actuated. 

In \cite{Krener1992}, the author claims that a pair $(\mPi,\mGamma)$ should be chosen that solves a proposed parametric optimization problem (OP) with constraints $\RG$. It was shown in \cite{Bernhard2018a} that the solution of this OP \textit{without} \eqref{eq_reg2} is connected to the optimal solution of an LQT problem. However, under constraint \eqref{eq_reg2}, it has not been proven yet that the obtained pair is optimal for every $\xol_0$. Nonetheless, this OP is frequently used in recent approaches, cf. \cite{Gao2016}, \cite[Sec. V-A]{Serrani2012}. Another parametric OP is introduced in the context of an optimal servo-compensator design in \cite{Moase2016}. The obtained $(\mPi\xol,\mGamma\xol)$ minimizes the power of the expected value of a cost similar to \eqref{eq_Ju} over a given distribution of $\xol_0$. Again, this implies by no means that the power is minimized for every $\xol_0$. In addition, the required computations by this approach are very involved, cf. \cite[Sec. 4.1]{Moase2016}. Without additional proofs, it cannot be conclude that either of the two approaches gives a suitable solution to our \oorref{oor}.

A different approach aims at deriving explicit degrees of freedom (DOF) that influence the state $\x$ but do not affect the output $\y$. Similarly to \textit{control allocation} (\cite{Haerkegard2005}), they are used for an online optimization during tracking. This idea originates from \cite{Zaccarian2009} and is generalized in \cite{Cristofaro2014}. With regard to output regulation, results are given in, e.g.,  \cite{Galeani2014}, \cite{Galeani2015} and \cite{Serrani2012}. Once these DOF are at hand, they shall be used beneficially. Then, one still faces an optimal tracking problem if a cost such as \eqref{eq_Ju} is considered. It is suggested in \cite{Cristofaro2014,Galeani2014} to solve such problems by an online and dynamic gradient descent as a part of a dynamic control strategy. For quadratic costs (e.g., \cite[Sec. VI]{Cristofaro2014}, \cite[Sec. V]{Galeani2014}), our proposed results could be used to avoid controller dynamics and the complications of making the DOF explicitly available. In \cite{Bernhard2017}, explicit DOF are obtained by a row-by-row decoupling control design and are used to explicitly solve an LQT problem with cost \eqref{eq_Ju}. If decoupling is not required, this causes unnecessary restrictions on the structure of the control and complications in its derivation.

In the literature, question \ref{ques2}) is usually considered in the context of the control design for under-actuated systems \eqref{eq_sys}.

Many contributions evaluate the achievable performance of regulating the output $\y$ either to an accessible reference $\yol$ (e.g., \cite{Chen2002}) or to the best $\ys$ with respect to the cost \eqref{eq_Jy} when $\Aol=\zero$ (e.g., \cite{Garcia2011}). For such constant $\yol$, \oorref{uor} is well understood since it is equivalent to a parametric OP stated in \cite{Willems2004}. This OP is also considered by \cite{Davison2011}.

For quasi-periodic references, \cite{Corona2018b} considers \oorref{uor} when the input energy \eqref{eq_Ju} with $\R\succ\zero$ is added to the cost \eqref{eq_Jy}. The calculation of the proposed control is rather involved and only valid for diagonal weights $\Q$, $\R$. The results may be extended to the case ${\R=\zero}$ which will however require additional assumptions. Our approach avoids these ambiguities and disadvantages.
 By using a cheap optimal control approach, \cite[Ch. 17]{Saberi2000} proves that a solution to \oorref{uor} exists. However, it is not shown how a solution pair $(\mPi^y,\mGamma^y)$ can actually be derived which is not obvious unfortunately. This is an essential part of our contribution.

To the best of our knowledge, we conclude that solutions to \oorref{oor} and \oorref{uor}, that hold under general assumptions and are derived as easily as in the case when \orref{or} has a unique solution, are not known yet.

\subsection{Outline}
The next section presents preliminaries that include definitions and basic assumptions, and we introduce a unifying linear quadratic tracking approach. Based on this, we derive our main results in \secref{sec_oor}, where we start with \oorref{oor}, proceed with \oorref{uor} and, eventually, bridge the gap to a classical infinite-horizon LQT problem. Before our final conclusions, we discuss extensions of our main results and additional findings in \secref{sec_disc}. For the reader's convenience, we shift some prior results and proofs to the appendix.

 \section{Preliminaries}
\label{sec_prelim}

\subsection{Mathematical Notations}
\label{sec_math}
The real part of a complex number $c$ is $\Re{c}$ and we write $c\in j\RZ$ if $\Re{c}=0$. The zero matrix $\zero$ and identity matrix $\I$ have appropriate dimensions.
A matrix $\M$ is symmetric positive (semi)definite if $\M\succ(\succeq)\;\zero$. Its transpose is $\M^\tp$ and its spectrum is $\sigma(\M)$.  The conjugate transpose of a complex matrix $\M\in\CZ^{a\times b}$ is $\M^\hp$. We define the nullspace by $\nulls(\M)=\{\vec{\nu}\in\CZ^{b}\,\vert\,\zero=\M\vec{\nu}\}$ and the left nullspace by $\text{leftnull}(\M)=\{\vec{\nu}\in\CZ^{a}\,\vert\,\zero=\vec{\nu}^\hp\M\}$. We denote the (induced) $2$-norm by $\Vert\cdot\Vert_2$ and the Frobenius norm by $\Vert\cdot\Vert_\text{F}$. We use the big $\mathcal{O}$ notation: $f(\epsilon)=\mathcal{O}(g(\epsilon))$ as $\epsilon \to 0$ if and only if $\exists \alpha, \epsilon_0>0$ such that $\vert f(\epsilon)\vert \leq \alpha\vert g(\epsilon)\vert$ if $\vert \epsilon \vert < \epsilon_0$. 
For the ease of presentation, functions $f(x,y,\ldots)$ are often abbreviated by $f\cd$ if their arguments $x,y,\ldots$ are clear from context. Also, the dependence of a variable $x(t)$ on its argument time $t$ is often dropped. A variation of a possibly vector-valued function $\x\cd$ is denoted by $\dx\cd$. The $i$-th variation is written as $\delta^i\x\cd$ with $i\in\NZ$. For a system $(\C,\A,\B,\D)$ given by \eqref{eq_sys}, the \textit{Rosenbrock system matrix} is%
\begin{equation*}
\Ros{s,\C,\A,\B,\D}\coloneqq\ma s\I-\A & -\B\\
\C & \D \me
\end{equation*}
where $s\in\CZ$. For $\D=\zero$, we abbreviate $\Ros{s,\C,\A,\B}$.

\subsection{Definitions and Basic Assumptions}
\label{sec_def}

The following definitions of admissible solutions and optimality in a stationary sense are important.
\begin{defn}[Admissible solutions]
	Every stationary solution of \eqref{eq_linsys}, that is every pair $\big(\xs(t),\us(t)\big)=(\mPi\xol(t),\mGamma\xol(t))$ where $(\mPi, \mGamma)\in\RZ^{n\times\nol}\times\RZ^{m\times\nol}$ solves \eqref{eq_reg1}, is admissible.
\end{defn}
\begin{defn}[Optimality] \label{def_opt}
		For a given $\xol_0$ and the power $P\big(\xs(\cdot),\us(\cdot)\big)\coloneqq\lim_{T\to\infty}\left(\frac{1}{T} J_T\big(\xs(\cdot),\us(\cdot)\big)\right)$ w.r.t. the cost $J_T(\cdot)$, an admissible solution $\big(\xs^*(\cdot),\us^*(\cdot)\big)$ is \textit{optimal} if for every admissible solution $\big(\xs(\cdot),\us(\cdot)\big)$, the difference of powers satisfies%
		\begin{equation}\label{eq_delta_power}
			\Delta P\cd\coloneqq P\big(\xs(\cdot),\us(\cdot)\big)-P\big(\xs^*(\cdot),\us^*(\cdot)\big)\geq 0.
		\end{equation}
		If $\Delta P(\cdot)>0$ is true for all $\big(\xs(\cdot),\us(\cdot)\big)\not\equiv\big(\xs^*(\cdot),\us^*(\cdot)\big)$, then $\big(\xs^*(\cdot),\us^*(\cdot)\big)$ is \textit{unique}.
\end{defn}

 \remark The choice of the powers $P^u(\cdot)$ and $P^y(\cdot)$ as performance indices is beneficial for a stationary analysis: Under standard assumptions, both are well-defined in contrast to the costs \eqref{eq_Ju} and \eqref{eq_Jy} for $T\to\infty$ (cf. \cite{Anderson2007,Artstein1985,Bernhard2017b}); the asymptotic transition to the stationary state is disregarded; and any difference in power implies a difference in cost that grows linearly with time $T$ which is a strong optimality property over infinite horizons. For these reasons, the power concept is used in \cite{Corona2018b,Moase2016} and \cite{Saberi2000}, too.
 \remark Due to the framework of output regulation theory, we only consider stationary solutions here. 
 However, this restriction is not substantial in our case as we will discuss in \secref{sec_var}. There we will discuss that our results satisfy a strong overtaking property \cite{Artstein1985} even if more general admissible solutions are considered.

Throughout this paper, our basic assumptions are:
\assumption The pair $(\A,\B)$ is stabilizable. \label{ass_stab}
\assumption For all $\lambdaol \in \sigma\left(\Aol\right)$ it holds $\Re{\lambdaol}=0$ and the algebraic and geometric multiplicities are equal.\label{ass_eig}%

By the first assumption, we may assume that the feedback $-\K\x$ in control \eqref{eq_u} stabilizes \eqref{eq_linsys} and that the state converges from all $\x_0$ to a desired stationary state $\xs(\cdot)$ on which we focus in the sequel. 

The second assumption is important. It is necessary for ensuring that references and disturbances are bounded for all $\xol_0$ which are thus constant and (quasi)periodic. Only then the powers $P^u(\cdot)$ and $P^y(\cdot)$ are bounded in general and it is easy to see that the limit defining them exists. 
Anyhow, asymptotically stable dynamics in \eqref{eq_exosys} are of no interest since they do not contribute to the stationary behavior, cf. \cite[Sec. 9.1]{Trentel2001}. To exclude polynomial and exponentially unstable dynamics is also reasonable as we will discuss in \secref{sec_unstable}. 
We emphasize that \assref{ass_eig} or even stricter assumptions are standard in the relevant literature, e.g., \cite{Corona2018b}, \cite{Galeani2015}, \cite{Krener1992}, \cite[Ch. 17]{Saberi2000} and \cite{Serrani2012}.%

\subsection{A Unifying Linear Quadratic Tracking Approach}
\label{sec_lqt}
In this section, we present a unifying linear quadratic tracking (LQT) approach.  
In \secref{sec_oor}, the techniques developed here will serve to find solution candidates to both, \oorref{oor} and \ref{uor}, and to prove the main results. We consider the
\begin{lqtprob}\label{lqt}
For a given $\xol_0$, find an admissible solution $\big(\xs^*(\cdot),\us^*(\cdot)\big)$ that minimizes $P(\cdot)=\lim_{T\to\infty}\frac{1}{T} J_T\cd$ w.r.t.%
	\begin{equation}\label{eq_Jlqt}
	J_T(\cdot)=\tfrac{1}{2}\int_0^T (\ys-\yol){^\tp}\rho\Qy(\ys-\yol)+\us{^\tp}\epsilon\R\us \dop t
	\end{equation}
	where $\rho>0$ and $\epsilon>0$ as wells as $\Q\succ \zero$ and $\R\succ\zero$.
\end{lqtprob}

We formulated \lqtref{lqt} in accordance with the problems in \secref{sec_contr}. For $\rho>0$ and $\epsilon>0$, a solution can be constructed as in, e.g., \cite{Bernhard2017b}, \cite{Kreindler1969}, and its optimality follows from the results in \cite{Bernhard2017b}, \cite{Artstein1985} under weak assumptions. It balances the error energy versus the input energy based on the cost \eqref{eq_Jlqt} and, thus, depends on the introduced parameters $\rho$, $\epsilon$. In \secref{sec_rellqt}, we will show that \oorref{oor} and \ref{uor} are special cases of \lqtref{lqt} as $\rho\rightarrow\infty$ and $\epsilon\rightarrow0$, respectively.

At this point, we want to obtain conditions that a solution candidate $(\xs^*\cd,\us^*\cd)$ to \lqtref{lqt} should satisfy. In this respect, we use the \textit{calculus of variations}, see \cite[Ch. 5]{Athans1966} and \cite[Ch. 2]{Bryson1975} for details on the developments of this section. These techniques will essentially help us to derive and to prove our main results.

For every admissible solution $(\xs,\us)$ and every $\gamma\in\RZ$, there exists a \textit{stationary variation} $(\dx,\du)\coloneqq(\dPi\xol,\dGamma\xol)$ where $(\dPi,\dGamma)$ satisfies $\dPi\Aol=\A\dPi+\B\dGamma$ such that $\xs(\cdot)=\xs^*(\cdot)+\gamma\dx(\cdot)$ and $\us(\cdot)=\us^*(\cdot)+\gamma\du(\cdot)$. This is easily verified since \eqref{eq_reg1} depends affinely on $\mPi$ and $\mGamma$.  
We emphasize that $\dx(0)\not=\zero$ in general. The $i$-th variation of a cost functional $J_T(\cdot)$ is defined by $$\dJ{i}_T(\dx,\du,\xs^*,\us^*)\coloneqq \tfrac{d^i J_T(\text{$\xs^*+\gamma\dx,\us^*+\gamma\du$})}{d\gamma^i}\Big\vert_{\gamma=0}.$$ Then $J_T\cd$ can be equivalently written as its Taylor series at ${\gamma=0}$ which reads
$J_T(\xs,\us)=J_T(\xs^*,\us^*)+\Delta J_T\cd$ with the \textit{cost difference} $\Delta J_T\cd$ given by
\begin{equation}\label{eq_diffJ}
\Delta J_T(\dx,\du,\gamma,\xs^*,\us^*)= \dJ{1}_T(\cdot)\gamma+\tfrac{1}{2}\dJ{2}_T(\cdot)\gamma^2.
\end{equation}
In the next sections, we will choose $\gamma = 1$ without loss of generality. 

In view of \eqref{eq_Jlqt}, we introduce the \textit{Hamiltonian} function
\begin{multline*}
H(\xs,\us,\phiv,\xol(t))\coloneqq\tfrac{1}{2}\left((\ys-\yol){^\tp}\rho\Qy(\ys-\yol)\right. \\
\left.+\us{^\tp}\epsilon\R\us\right)+\phiv^\tp\left(\A\xs+\B\us+\Ed \xol\right)
\end{multline*}
where the \textit{costate} $\phiv(t)\colon[0,\infty)\rightarrow\RZ^n$ is some arbitrary function for now.  By using integration by parts, we are able to rewrite \eqref{eq_Jlqt}: 
\begin{multline*}
\hspace{-2mm}J_T(\xs\cd,\us\cd)=\int_0^T \left(H(\xs,\us,\phiv,\xol)-\phiv^\tp\dot\x_\text{s}\right) \dop t
\\=\int_0^T \left(H(\cdot)+\dot\phiv^\tp\xs\right) \dop t+\phiv(0)^\tp\xs(0)-\phiv(T)^\tp\xs(T)\hspace{-2mm}
\end{multline*}
On the basis of this form of $J_T\cd$, the first variation of $J_T\cd$ in \eqref{eq_diffJ} is calculated:
\begin{multline*}
\hspace{-3mm}\dJ{1}_T\cd=\int_0^T \hspace{-1mm}\left(\tfrac{\partial H(\txtfor{\xs,\us,\cdot})}{\partial \txtfor{\xs}}\Big\vert_*\hspace{-1mm}+\dot\phiv \right)^\tp\hspace{-1mm}\dx+\tfrac{\partial H(\txtfor{\xs,\us,\cdot})}{\partial \txtfor{\us}}^\tp\Big\vert_*\du \dop t \\+\phiv(0)^\tp\dx(0)-\phiv(T)^\tp\dx(T)\textkom
\end{multline*}
where we used the notation: $f(\xs,\us,\cdot)\vert_*\coloneqq f(\xs^*,\us^*,\cdot)$. 
The second variation of $J_T\cd$ is directly obtained from \eqref{eq_Jlqt}:
\begin{equation}\label{eq_2dJ}
\dJ{2}_T\cd=\int_0^T \dx^\tp\C^\tp\rho\Qy\C\dx+\du^\tp\epsilon\R \du \dop t.
\end{equation}
As we see, only the first variation $\dJ{1}_T\cd$ depends on the candidate $\big(\xs^*\cd,\us^*\cd\big)$ and the costate $\phiv\cd$. 
Hence, we may choose $\big(\xs^*\cd,\us^*\cd\big)$ and $\phiv\cd$ such that
\begin{subequations}\label{eq_cond}\begin{align}
	0&=\left(\tfrac{\partial H(\txtfor{\xs,\us,\phiv,\xol})}{\partial \txtfor{\xs}}\Big\vert_*+\dot\phiv \right)^\tp\dx\textkom \label{eq_Hx}\\
	0&=\tfrac{\partial H(\txtfor{\xs,\us,\phiv,\xol})}{\partial \txtfor{\us}}^\tp\Big\vert_*\du \label{eq_Hu}
\end{align}\end{subequations}
hold for all $\big(\dx\cd,\du\cd\big)\not\equiv(\zero,\zero)$ and some $\phiv(0)\in\RZ^n$. Since neither an initial value $\xs(0)$ is given nor any transversality condition is available, $\dJ{1}_T\cd$ does not vanish by using \eqref{eq_cond}, but it reduces to%
\begin{equation} \label{eq_1dJ}
\hspace{-3mm}\dJ{1}_T\cd=\phiv(0)^\tp\dx(0)-\phiv(T)^\tp\dx(T).
\end{equation}
The coefficients of $\dx$ and $\du$ in \eqref{eq_cond} vanish if we choose the costate dynamics $\dot{\phiv}=-\A^\tp\phiv-\C^\tp\rho\Qy(\ys-\yol)$ and $\epsilon\R\us^*+\B^\tp\phiv=\zero$ as a constraint. The costate dynamics together with \eqref{eq_linsys} are called the \textit{Hamiltonian} system. Following \cite{Bernhard2018a} and \cite{Bernhard2017b}, the costate dynamics and this constraint are satisfied for $\big(\xs^*,\us^*)=(\mPi\xol,\mGamma\xol)$ and the approach $\phiv\cd=\mPic\xol\cd$ if the triple $(\mPi,\mPic,\mGamma)$ solves \eqref{eq_reg1} and
\begin{subequations}\label{eq_nec}\begin{align}
	\mPic\Aol&=-\A^\tp\mPic-\C^\tp\rho\Qy(\C\mPi+\D_d-\Col)\textkom \label{eq_necc}\\
	\zero&=\epsilon\R\mGamma+\B^\tp\mPic. \label{eq_necu}
	\end{align}\end{subequations}
In view of \cite{Bernhard2017b}, $\big(\xs^*,\us^*)$ can also be derived by the \textit{minimum principle} for an infinite horizon  \cite{Halkin1974}. Hence, $(\mPi\xol,\mGamma\xol)$ qualifies as an optimal solution if general variations are considered which we discuss in \secref{sec_var}. Since we only consider stationary variations at first, the conditions \eqref{eq_cond} may not be necessary for optimality. Nonetheless, choosing $(\mPi\xol,\mGamma\xol)$ in order to satisfy \eqref{eq_cond} is favorable as we show:
\begin{lemma}\label{lem_delta1J}
	Suppose both conditions \eqref{eq_Hx} and \eqref{eq_Hu} hold for an admissible solution $(\xs^*,\us^*)=(\mPi\xol,\mGamma\xol)$ and the costate $\phiv=\mPic\xol$. Then, for every admissible solution $(\xs,\us)$, the power difference satisfies $\Delta P(\xs,\us,\xs^*,\us^*)=\lim_{T\to\infty}\frac{1}{2T}\dJ{2}(\dx,\du)$.
\end{lemma}
\begin{proof}
	The cost difference $\Delta J_T\cd$ is given by \eqref{eq_diffJ} and, thus, it results $\Delta P\cd=\lim_{T\to\infty}\frac{1}{T}\big(\dJ{1}_T(\cdot)+\tfrac{1}{2}\dJ{2}_T(\cdot)\big)$ for $\gamma=1$. By \eqref{eq_Hx} and \eqref{eq_Hu}, \eqref{eq_1dJ} is true.  
	Due to \assref{ass_eig}, $\Vert \phiv(t)\Vert_2$, $\Vert \dx(t)\Vert_2$ and, hence, $\vert\dJ{1}_T(\cdot)\vert$ are bounded functions on $[0,\infty)$ for all $\xol_0$ which implies $\lim_{T\to\infty}\frac{1}{T}\vert\dJ{1}_T(\cdot)\vert=0$.
\end{proof}

Due to this lemma, the optimality analysis of such $(\xs^*,\us^*)$ is very promising. In the proofs of our main results, we exploit this and the quadratic nature of $\dJ{2}(\dx,\du)$ in \eqref{eq_2dJ} by applying the useful \lemref{lem_Wo} given in \appref{app_A}.%

 \section{Optimal Output Regulation for\\ Linear Systems}
\label{sec_oor}

In this section, we derive our main results that are two new sets of regulator equations which provide solutions to the optimal output regulation problems \ref{oor} and \ref{uor}, respectively. We investigate solvability conditions, the connection to the classical regulator equations $\RG$ and how \lqtref{lqt} unifies both problems.

\subsection{If Regulator Equations $\RG$ Have Infinitely Many Solutions}
\label{sec_oa}

\ldots\,then we want to answer question \ref{ques1}), i.e., we seek an optimal solution $(\xs^*,\us^*)=(\mPi^\oa\xol,\mGamma^\oa\xol)$ to \oorref{oor}. This is generally desired for over-actuated systems: $\rank{\B}>\rank{\C}$.

To find a suitable candidate $(\mPi^\oa\xol,\mGamma^\oa\xol)$, we carry out the analysis in \secref{sec_lqt} with respect to \oorref{oor}. First, we note that \oorref{oor} is equivalent to \lqtref{lqt} with additional state constraint \eqref{eq_constr} and for the choice $\epsilon=1$. Since $\ys(t)-\yol(t)\equiv\zero$, an admissible solution $(\mPi\xol,\mGamma\xol)$ is feasible only if $(\mPi,\mGamma)$ solves the regulator equations $\RG$, see \cite{Francis1977}. Hence, a variation $(\dx,\du)=(\dPi\xol,\dGamma\xol)$ defined by $\dPi=\mPi-\mPi^\oa$ and $\dGamma=\mGamma-\mGamma^\oa$ (choose $\gamma=1$) must also satisfy $\C\dx(t)=\C\dPi\xol(t)\equiv\zero$ $\forall \xol_0\neq\zero$ and, accordingly,%
\begin{subequations}\label{eq_feasvaroa}\begin{align}
	\dPi\Aol&=\A\dPi+\B\dGamma,\label{eq_statvar}\\
	\zero&=\C\dPi. \label{eq_var_constr}
	\end{align}\end{subequations}
We remark that a nontrivial solution of \eqref{eq_feasvaroa} exists if and only if the solution of $\RG$ is not unique.
Next, we construct a candidate $(\mPi^\oa\xol,\mGamma^\oa\xol)$ such that \lemref{lem_delta1J} can be applied, i.e., conditions \eqref{eq_Hx} and \eqref{eq_Hu} hold. In \secref{sec_lqt}, we had to regard arbitrary stationary variations and, thus, we chose the coefficient of $\dx$ in \eqref{eq_Hx} equal zero. However, here we may require instead
\begin{equation}\label{eq_Hx_oa}
\tfrac{\partial H(\txtfor{\xs,\us,\phiv,\xol})}{\partial \txtfor{\xs}}\Big\vert_*+\dot\phiv=\C^\tp\mGammac^\oa\xol
\end{equation}
with arbitrary $\mGammac^\oa\in\RZ^{p\times\nol}$. Then \eqref{eq_Hx} equals $\xol^\tp\mGammac^{\oa\tp}\C\dx=\zero$ which is always true since $\C\dx(t)\equiv\zero$. Since $\ys\equiv\yol$, it follows ${H\cd=\tfrac{1}{2}\us^\tp\R\us+\phiv^\tp(\A\xs+\B\us+\Ed)}$. Then to satisfy \eqref{eq_Hx_oa},
we choose ${\phiv=\mPic^\oa\xol}$ and obtain $\mPic^\oa\Aol=-\A^\tp\mPic^\oa+\C^\tp\mGammac^\oa$. Though $(\dPi,\dGamma)$ is now constrained by \eqref{eq_feasvaroa}, we still choose \eqref{eq_necu} to satisfy \eqref{eq_Hu}. For $\epsilon=1$, this gives $\mGamma^\oa=-\R^{-1}\B^\tp\mPic^\oa$.

Taking these equations together with $\RG$ into account, we are able to present the \textit{new regulator equations} $\RGoa$:%
\begin{subequations}\label{eq_regoa}\begin{align}[left = \RGoa\empheqlbrace]
	\mPic^\oa\Aol&=-\A^\tp\mPic^\oa+\C^\tp\mGammac^\oa \label{eq_oa1}\\
	\mPi^\oa\Aol&=\A\mPi^\oa-\B\R^{-1}\B^\tp\mPic^\oa+\Ed \label{eq_oa2}\\
	\Col&=\C\mPi^\oa+\Dd. \label{eq_oa3}
\end{align}\end{subequations}
A solution $(\mPi^\oa,\mPic^\oa,\mGammac^\oa)$ of $\RGoa$ provides our solution candidate $(\mPi^\oa,-\R^{-1}\B^\tp\mPic^\oa)$ to \oorref{oor}. 
Before we verify its optimality, we state the powerful result that we can always solve $\RGoa$ if \orref{or} has at least one solution.%
\begin{lemma}\label{lem_solv_regoa}
	A triple $(\mPi^\oa,\mPic^\oa,\mGammac^\oa)$ solving the new regulator equations $\RGoa$ exists if and only if a pair $(\mPi,\mGamma)$ solving the classical regulator equations $\RG$ exists. \hfill\QED
\end{lemma}
The proof is given in \appref{app_B}. Now, we are ready to derive our \textbf{first main result}:
\begin{theorem}\label{thm_oa}
	A solution to \oorref{oor} exists if and only if a triple $(\mPi^\oa,\mPic^\oa,\mGammac^\oa)$ exists that solves the new regulator equations $\RGoa$. The optimal solution is $(\xs^*,\us^*)=(\mPi^\oa\xol,\mGamma^\oa\xol)$ where $\mGamma^\oa=-\R^{-1}\B^\tp\mPic^\oa$, which minimizes $P^u(\cdot)$ under constraint \eqref{eq_constr} for every $\xol_0$. It is unique if and only if the following condition holds:
	\begin{equation}\label{eq_obsv_cond}
		\rank{\ma \lambdaol\I-\A^\tp &  \C^\tp\me}=n\textkom\;\; \forall \lambdaol\in\sigma\left(\Aol\right).
	\end{equation}
\end{theorem}
\begin{proof}
	The solvability of \oorref{oor} requires that the constraint~\eqref{eq_constr} is satisfied for some $(\mPi\xol,\mGamma\xol)$ and all $\xol_0$. This implies that $\RG$ has a solution under the present assumptions \ref{ass_eig}, cf. \cite{Francis1977}. By \lemref{lem_solv_regoa}, $\RGoa$ has a solution which proves the only if.
	
	Following the discussion at the beginning of this section, for a candidate $(\xs^*,\us^*)=(\mPi^\oa\xol,\mGamma^\oa\xol)$ and $\phiv=\mPic^\oa\xol$ obtained by $\RGoa$ all conditions of \lemref{lem_delta1J} hold. Hence, for every feasible variation obtained from \eqref{eq_feasvaroa} it results%
	\begin{equation}\label{eq_diffPu}
	\hspace{-4mm}\Delta P^u(\dx,\du,\xs^*,\us^*)=\lim_{T\to\infty}\frac{1}{2T}\int_0^T\hspace{-2mm} \xol(t)^\tp\dGamma^\tp\R \dGamma\xol(t) \dop t.\hspace{-2mm}
	\end{equation}
	Since $\R\succ\zero$, we find $\Delta P^u\cd\geq\zero$ for all feasible $(\dx,\du)$ and all $\xol_0$. Thus, $(\mPi^\oa\xol,\mGamma^\oa\xol)$ is an optimal solution.
	
	In order to verify uniqueness if \eqref{eq_obsv_cond} holds,  we show by an exhaustive three-part case study that $\Delta P^u\cd>0$ is satisfied for all nontrivial variations $(\dPi\xol,\dGamma\xol)\not\equiv(\zero,\zero)$. As a consequence of \lemref{lem_Wo}, $\Delta P^u\cd>0$ is true for all nontrivial variations for which the system  $(\R_{\text{\tiny$\nicefrac{1}{2}$}}\dGamma,\Aol)$ is completely observable (for some $\R_{\text{\tiny$\nicefrac{1}{2}$}}$: $\R=\R_{\text{\tiny$\nicefrac{1}{2}$}}^\tp\R_{\text{\tiny$\nicefrac{1}{2}$}}$). Suppose instead that not all but some eigenvalues of $\Aol$ are observable and assume that the system $(\R_{\text{\tiny$\nicefrac{1}{2}$}}\dGamma,\Aol)$ is given in form of the decomposition~\eqref{eq_exo_decomp}. Since $\R_{\text{\tiny$\nicefrac{1}{2}$}}$ is invertible, it results $\dGamma=\ma \dGamma_1 & \zero \me$, i.e., $\dGamma_2=\zero$. Based on \eqref{eq_feasvaroa}, it follows that $\dPi_2\Aol_{22}=\A\dPi_2$ and $\zero=\C\dPi_2$. Because $\Aol_{22}$ is diagonal, for each column $\delta\vec{\pi}_{2,i}$ of $\dPi_2$ it must hold $\delta\vec{\pi}_{2,i}^\tp \ma \lambdaol\I-\A^\tp &  \C^\tp\me=\zero$ for its associated $\lambdaol\in\sigma(\Aol_{22})$. By condition \eqref{eq_obsv_cond}, it results $\delta\vec{\pi}_{2,i}=\zero$ for all columns which leads to $\dPi_2=\zero$. Hence, every such variation reads $(\dPi\xol,\dGamma\xol)=(\dPi_1\xol_1,\dGamma_1\xol_1)$. It is nontrivial only if $\xol_1(0)\neq\zero$ which implies $\Delta P^u\cd>0$ by \lemref{lem_Wo}. Completing the case study, we notice that any variation for which $(\R_{\text{\tiny$\nicefrac{1}{2}$}}\dGamma,\Aol)$ is completely unobservable implies $\dGamma=\zero$ and, hence, is trivial since $\dPi=\zero$ due to \eqref{eq_obsv_cond}.
	
	To show necessity of condition \eqref{eq_obsv_cond}, suppose $(\xs^*,\us^*)$ is unique and \eqref{eq_obsv_cond} does not hold for some $\lambdaol\in\sigma(\Aol)$. Hence, for some $\CZ^{n}\ni\sdPi\neq\zero$ it results $\sdPi^\tp \ma \lambdaol\I-\A^\tp &  \C^\tp\me=\zero$. Due to \assref{ass_eig}, we suppose that $\Aol$ is diagonal without loss of generality and the above holds for $\lambdaol$ being the first element on the diagonal. Then we construct the feasible variation $(\sdPi\overline{x}_1,\zero)$ that is nontrivial if $\overline{x}_1(0)\neq0$ and for which $\Delta P^u\cd=0$ holds which contradicts uniqueness.
\end{proof}	

To guarantee the \textit{solvability} of $\RGoa$, we consider the well known \textit{non-resonance condition} \cite{Isidori2017}:%
\assumption \label{ass_nonres} For all $\lambdaol\in\sigma\left(\Aol\right)$, it holds
\begin{equation}\label{eq_nonres}
	\rank{\Ros{\lambdaol,\C,\A,\B}}=n+p.
\end{equation}

Actually, this condition is true for most over-actuated systems if $\rank{\B}=m>p=\rank{\C}$ since these systems usually do not have any invariant zeros, see \cite[Thm. 5]{Davison1974}.  
By \cite[Thm. 1.9]{Huang2004}, we recall that $\RG$ is solvable for all $\Ed$, $\Dd$ and $\Col$ if and only if \assref{ass_nonres} holds. Hence, the next result is immediate due to \lemref{lem_solv_regoa}:%
\begin{coroll}\label{coroll_nonres_oa}
	The new regulator equations $\RGoa$ have a solution $(\mPi^\oa,\mPic^\oa,\mGammac^\oa)$ for all matrices $\Ed$, $\Dd$ and $\Col$ if and only if \assref{ass_nonres} holds. This solution is unique if and only if condition \eqref{eq_obsv_cond} is satisfied in addition.%
	  \hfill\QED
\end{coroll}

\remark \label{rem_DJ_oa}
 When condition \eqref{eq_obsv_cond} does not hold, one may wonder if some $(\mPi^\oa\xol,\mGamma^\oa\xol)$ among the optimal stationary solutions exists that performs better with respect to $J_T^u(\cdot)$ (since $\Delta P^u(\cdot)=0$ for all of them). Following the proof of \thmref{thm_oa}, all these are of the form $\big((\mPi^\oa+\dPi)\xol,\mGamma^\oa\xol\big)$ and, clearly, it holds $\Delta J_T^u(\cdot)=0$.

\subsection{If Regulator Equations $\RG$ Have No Solution At All}
\label{sec_ua}

\ldots\,then we want to answer question \ref{ques2}) and seek a solution $(\xs^*,\us^*)=(\mPi^\ua\xol,\mGamma^\ua\xol)$ to \oorref{uor} which is very important for under-actuated systems, i.e., for $\rank{\B}<\rank{\C}$.

Comparing \eqref{eq_Jy} with \eqref{eq_Jlqt} for a given pair $(\mPi,\mGamma)$, both costs coincide for all $\xol_0$ if we choose $\rho=1$ and $\epsilon=0$. Thus, to carry out the analysis in \secref{sec_prelim} for \oorref{uor} is equivalent to
the substitution of $\rho=1$ and $\epsilon=0$ in \eqref{eq_necc} and \eqref{eq_necu}. Accordingly, we derive the \textit{new regulator equations} $\RGua$:%
\begin{subequations}\label{eq_regua}\begin{align}[left = \hspace{-3.5mm}\RGua\empheqlbrace]
	\mPi^\ua\Aol&=\A\mPi^\ua+\B\mGamma^\ua+\Ed \label{eq_ua12}\\
	\mPic^\ua\Aol&=-\A^\tp\mPic^\ua-\C^\tp\Qy(\C\mPi^\ua+\D_d-\Col)\hspace{-1.5mm} \label{eq_ua2}\\
	\zero&=-\B^\tp\mPic^\ua. \label{eq_ua3}
	\end{align}\end{subequations}
For the same reasons as in the discussion preceding \lemref{lem_delta1J}, a solution to \oorref{uor} may not necessarily satisfy $\RGua$. However, to the best of our knowledge, necessary optimality conditions are not available for \oorref{uor}. Hence, we focus on showing that if a solution of $\RGua$ exists, then $(\xs^*,\us^*)=(\mPi^\ua\xol,\mGamma^\ua\xol)$
is an optimal solution among all stationary solutions. These can always be written as $(\xs,\us)=(\xs^*+\dPi\xol,\us^*+\dGamma\xol)$ where $(\dPi,\dGamma)$ has to solve \eqref{eq_statvar}. \textit{Solvability} of \eqref{eq_regua} is shown under the very general
\assumption 
\label{ass_nonres_ua} For all $\lambdaol\in\sigma\left(\Aol\right)$, it holds
\begin{equation}\label{eq_nonres_ua}
\rank{\Ros{\lambdaol,\C,\A,\B}}=n+m.
\end{equation}

\begin{lemma}\label{lem_nonres_ua}
	The new regulator equations $\RGua$ have a unique solution $(\mPi^\ua,\mPic^\ua,\mGamma^\ua)$ for all matrices $\Ed$, $\Dd$ and $\Col$ if and only if \assref{ass_nonres_ua} holds. \hfill\QED
\end{lemma}

We prove this in \appref{app_B}. Remarkably, \assref{ass_nonres_ua} is true for most under-actuated systems if $\rank{\B}=m<p=\rank{\C}$ since these systems usually do not have any invariant zeros, see \cite{Davison1974}.
Hence, by considering \oorref{uor} for the most important case of under-actuated systems, it is reasonable to assume that $\RGua$ is solvable. In this light, we state our \textbf{second main result}:
\begin{theorem}\label{thm_ua}
	Suppose the triple $(\mPi^\ua,\mPic^\ua,\mGamma^\ua)$ solves the new regulator equations $\RGua$. Then the pair $(\mPi^\ua,\mGamma^\ua)$ solves \oorref{uor} and $(\xs^*,\us^*)=(\mPi^\ua\xol,\mGamma^\ua\xol)$ is an optimal solution that minimizes $P^y(\cdot)$ for every $\xol_0$. It is unique  if and only if \assref{ass_nonres_ua} holds.
\end{theorem}
\begin{proof}
	The proof is based on similar arguments as the proof of \thmref{thm_oa}.
	By construction, our candidate $(\xs^*,\us^*)=(\mPi^\ua\xol,\mGamma^\ua\xol)$ and $\phiv=\mPic^\ua\xol$ solves \eqref{eq_reg1} and \eqref{eq_nec} for $\rho=1$ and $\epsilon=0$. Thus, we can apply \lemref{lem_delta1J} and find for every variation given by \eqref{eq_statvar} that
	\begin{equation}\label{eq_diffPy}
	\Delta P^y\cd=\lim_{T\to\infty}\frac{1}{2T}\int_0^T \xol(t)^\tp\dPi^\tp\C^\tp\Q \C\dPi\xol(t) \dop t.
	\end{equation}
	Then we note that $\Delta P^y\cd\geq\zero$ for all feasible $(\dx,\du)$ and every $\xol_0$ since $\Q\succ\zero$, and $(\mPi^\ua\xol,\mGamma^\ua\xol)$ is optimal.
	
	We verify uniqueness if \assref{ass_nonres_ua} holds by a three-part case study.
	Consider every $(\dPi,\dGamma)$ given by \eqref{eq_statvar} such that the system $(\Q_{\text{\tiny$\nicefrac{1}{2}$}}\C\dPi,\Aol)$ is completely observable (for some $\Q_{\text{\tiny$\nicefrac{1}{2}$}}$ such that $\Q=\Q_{\text{\tiny$\nicefrac{1}{2}$}}^\tp\Q_{\text{\tiny$\nicefrac{1}{2}$}}$). It follows $\Delta P^y\cd>0$ for all $\xol\neq\zero$ as a consequence of \lemref{lem_Wo}. If $(\Q_{\text{\tiny$\nicefrac{1}{2}$}}\C\dPi,\Aol)$ is  not completely observable instead, we may assume that the system is decomposed as in \eqref{eq_exo_decomp}. An analogous decomposition of $\dPi$ gives $\G_1=\Q_{\text{\tiny$\nicefrac{1}{2}$}}\C\dPi_1$ and $\C\dPi_2=\zero$ since $\Q_{\text{\tiny$\nicefrac{1}{2}$}}$ is invertible. Hence, $(\dPi_2,\dGamma_2)$ has to satisfy $\dPi_2\Aol_{22}=\A\dPi_2+\B\dGamma_2$ and $\zero=\C\dPi_2$. By \assref{ass_nonres_ua}, we find that $\Ros{\lambdaol,\C,\A,\B}\vec{\nu}=\zero$ implies $\CZ^{n+m}\ni\vec{\nu}=\zero$ for all $\lambdaol\in\sigma(\Aol_{22})$. Because $\Aol_{22}$ is diagonal, both together implies that all columns of $\dPi_2$ and $\dGamma_2$ must vanish, i.e., $(\dPi_2,\dGamma_2)=(\zero,\zero)$. Hence, every such feasible variation $(\dPi_1\xol_1,\dGamma_1\xol_1)\not\equiv(\zero,\zero)$ is nontrivial only if $\xol_1(0)\neq\zero$ for which it always results $\Delta P^y(\cdot)>0$ based on \lemref{lem_Wo}. Eventually, we observe that a nontrivial variation for which $(\Q_{\text{\tiny$\nicefrac{1}{2}$}}\dPi,\Aol)$ is completely unobservable does not exist since $\C\dPi=\zero$ implies $(\dPi,\dGamma)=(\zero,\zero)$ by the analysis above. This completes the exhaustive case study and shows that $\Delta P^y(\cdot)>0$ for all nontrivial variations.
	
	Regarding necessity of \assref{ass_nonres_ua}, we suppose \eqref{eq_nonres_ua} is not satisfied for some $\lambdaol\in\sigma(\Aol)$ and $(\xs^*,\us^*)$ is unique. Then we find some $\CZ^{n+m}\ni\vec{\nu}\neq\zero$ such that $\Ros{\lambdaol,\C,\A,\B}\vec{\nu}=\zero$ from which we obtain $\big[ \sdPi^\tp \;\; \delta\vec{\gamma}^\tp \,\big]^\tp=\vec{\nu}$.  
	Without loss of generality, we assume that $\Aol$ is diagonal and $\lambdaol$ is the first element on the diagonal. Hence, we can construct the feasible variation $(\sdPi\overline{x}_1,\delta\vec{\gamma}\overline{x}_1)$ that is nontrivial if $\overline{x}_1(0)\neq0$ and for which $\Delta P^y\cd=0$ holds since $\C\sdPi\overline{x}_1(t)\equiv\zero$. This contradicts the uniqueness.
\end{proof}	

\remark \label{rem_DJ_ua}
Based on this proof, $\Delta P^y(\cdot)=0$ occurs only for stationary solutions $(\xs,\us)=(\xs^*+\dPi\xol,\us^*+\dGamma\xol)$ where $\C\dPi\xol\equiv\zero$ for the given $\xol_0$. Referring to \remref{rem_DJ_oa}, we easily see that then $\Delta J_T^y(\cdot)=0$ holds for all optimal solutions, too.
\remark \label{rem_alloc_ua}
\assref{ass_nonres_ua} requires that $\B$ has full rank. If instead $\rank{\B}<m$, without loss of generality, we equivalently rewrite system \eqref{eq_sys} with virtual inputs such that the associated new input matrix has full rank, apply \thmref{thm_ua} and use techniques of, e.g., \cite{Haerkegard2005}, to allocate the original control inputs.%

\subsection{Relation to Classical Infinite-Horizon LQ Tracking Problems}
\label{sec_rellqt}

In this section, we investigate that \oorref{oor} and \oorref{uor} are special cases of \lqtref{lqt}, and put them as well as \orref{or} in a greater context. This insight is useful when it suffices to satisfy the constraint \eqref{eq_constr} with (potentially arbitrarily) small errors or when the input energy should be considered in cost \eqref{eq_Jy} by a (potentially arbitrarily) small weight.

To this end, the cost \eqref{eq_Jlqt} in \lqtref{lqt} (for $\epsilon=1$) takes \eqref{eq_constr} implicitly into account by adding the error energy to \eqref{eq_Ju} which corresponds to an ``integral penalty function'', see \cite[Sec. 3.4]{Bryson1975}. Forcing the error to zero as ${\rho\rightarrow\infty}$ gives \oorref{oor} as a special case of \lqtref{lqt}.
 
Similarly, \oorref{uor} results from \lqtref{lqt} (for $\rho=1$) as ${\epsilon\rightarrow0}$. As in \cite[Ch. 17]{Saberi2000} , this cheap optimal control problem cannot be treated by known approaches (such as \cite{Anderson2007} 
and \cite{Pratti2008}) because the augmented system composed of \eqref{eq_linsys} and \eqref{eq_exo} is not stabilizable.

We consider a typical assumption in optimal tracking:
\assumption\label{ass_detect} The pair $(\C,\A)$ is detectable.

Let $\Q=\Q_{\text{\tiny$\nicefrac{1}{2}$}}^\tp\Q_{\text{\tiny$\nicefrac{1}{2}$}}\succ\zero$, then $(\Q_{\text{\tiny$\nicefrac{1}{2}$}}\C,\A)$ is detectable.

Following the discussion, we present our \textbf{third main result}.

\begin{theorem}\label{thm_lqt}
	Suppose that \assref{ass_detect} holds. It exists a unique optimal solution $\big(\mPi(\epsilon,\rho)\xol,\mGamma(\epsilon,\rho)\xol\big)$ to \lqtref{lqt} for every $\epsilon>0$ and $\rho>0$.
	\renewcommand{\theenumi}{\alph{enumi}}
	\begin{enumerate}
		\item \label{thm3_1} If \assref{ass_nonres} and the condition \eqref{eq_obsv_cond} are satisfied, then, as $\rho\to\infty$, $\Vert\mPi(1,\rho)-\mPi^\oa\Vert_F=\mathcal{O}(\nicefrac{1}{\rho})$ and $\Vert\mGamma(1,\rho)-\mGamma^\oa\Vert_F=\mathcal{O}(\nicefrac{1}{\rho})$ where $(\mPi^\oa,\mGamma^\oa)$ uniquely solves \oorref{oor}.
		\item \label{thm3_2} If \assref{ass_nonres_ua} holds, then $\Vert\mPi(\epsilon,1)-\mPi^\ua\Vert_F=\mathcal{O}(\epsilon)$ and $\Vert\mGamma(\epsilon,1)-\mGamma^\ua\Vert_F=\mathcal{O}(\epsilon)$ as $\epsilon\to0$ where $(\mPi^\ua,\mGamma^\ua)$ uniquely solves \oorref{uor}.
	\end{enumerate}
\end{theorem}
\begin{proof}
	Under \assref{ass_stab}, \ref{ass_eig} and \ref{ass_detect}, it exists a unique triple $\big(\mPi(\epsilon,\rho),\mPic(\epsilon,\rho),\mGamma(\epsilon,\rho)\big)$ that solves the system of equations \eqref{eq_reg1} and \eqref{eq_nec} for all $\epsilon, \rho>0$ (see \cite[Thm. 1]{Bernhard2018a} for a proof). From this, we obtain the unique optimal solution $\big(\mPi(\epsilon,\rho)\xol,\mGamma(\epsilon,\rho)\xol\big)$ to \lqtref{lqt}. Basically implied by \cite{Bernhard2017b}, it can be proven similarly as \thmref{thm_oa}. 
	
	We consider \ref{thm3_1}), i.e., \lqtref{lqt} for $\epsilon=1$. By \eqref{eq_necu}, we replace $\mGamma(1,\rho)=-\R^{-1}\B^\tp\mPic(1,\rho)$ in \eqref{eq_reg1}. The resulting equation and \eqref{eq_necc} can be equivalently rewritten by introducing an auxiliary variable $\mGammac(\rho)\colon(0,\infty)\rightarrow\RZ^{p\times\nol}$ (dropping the arguments):%
	\begin{subequations} \label{eq_reg_oa_dist}\begin{align}
		\mPic\Aol&=-\A^\tp\mPic+\C^\tp\mGammac\\
		\mPi\Aol&=\A\mPi-\B\R^{-1}\B^\tp\mPic+\Ed \\
		\Col&=\C\mPi+\D_d+\tfrac{1}{\rho}\Q^{-1}\mGammac. \label{eq_output_sub}
		\end{align}\end{subequations}
	Clearly, the set \eqref{eq_reg_oa_dist} results from $\RGoa$ by disturbing \eqref{eq_oa3} by $\nicefrac{1}{\rho}\Q^{-1}\mGammac(\rho)$. Hence, we are able to apply \lemref{lem_Axb} which proves \ref{thm3_1}).  
	In view of \secref{sec_ua}, we regard the two sets: $\RGua$ and \eqref{eq_reg1}, \eqref{eq_necc} for $\rho=1$, \eqref{eq_necu} which is disturbed by $\epsilon\R\mGamma(\epsilon,1)$. Again, we can apply \lemref{lem_Axb} that proves \ref{thm3_2}).%
\end{proof}

 \section{Discussion}
\label{sec_disc}

In this section, we shortly discuss extensions of our main results, e.g., for systems with feedthrough or for general admissible solutions.

\subsection{Feedthrough and Additional State Penalty}
\label{sec_feedthrough}

So far, we disregarded a feedthrough: $\y=\C\x+\D\u+\Dd\xol$ for the purpose of a concise presentation. When $\D\neq\zero$, then we modify $\RGua$ by replacing \eqref{eq_ua2} and \eqref{eq_ua3} by:
\begin{align*}
\mPic^\ua\Aol&=-\A^\tp\mPic^\ua-\C^\tp\Qy(\C\mPi^\ua+\D\mGamma^\ua+\D_d-\Col)\textkom\\
\zero&=-\B^\tp\mPic^\ua-\D^\tp\Q(\C\mPi^\ua+\D\mGamma^\ua+\D_d-\Col).
\end{align*}
Occasionally, contributions like \cite{Krener1992,Serrani2012} consider an additional state penalty in \oorref{oor} by adding $\nicefrac{1}{2}\int_0^T \xs^\tp\Q_x\xs\dop t$ with $\Q_x\succeq\zero$ to cost \eqref{eq_Ju}.  
Then, we modify $\RGoa$ by replacing \eqref{eq_regoa} completely by
\begin{subequations}
	\begin{align}
	\mPic^\oa\Aol&=-\Q_x\mPi^\oa-\A^\tp\mPic^\oa+\C^\tp\mGammac^\oa\textkom \\
	\mPi^\oa\Aol&=\A\mPi^\oa+\B\mGamma^\oa+\Ed\textkom \label{eq_RG_D1}\\
	\Col&=\C\mPi^\oa+\D\mGamma^\oa+\Dd \label{eq_RG_D2}
	\end{align}
\end{subequations}
where it must hold: $\mGamma^\oa=\R^{-1}(-\B^\tp\mPic^\oa+\D^\tp\mGammac^\oa)$. We also substitute $\Ros{\lambdaol,\C,\A,\B}$ with $\Ros{\lambdaol,\C,\A,\B,\D}$ in \assref{ass_nonres} and~\ref{ass_nonres_ua}. With respect to the modified two sets of equations, all results on solvability (for every choice of $\Q_x\succeq\zero$) and on optimality of $(\mPi^\ua,\mGamma^\ua)$ and $(\mPi^\oa,\mGamma^\oa)$ in \secref{sec_oor} hold. Especially \lemref{lem_solv_regoa} holds, where $\RG$ for $\D\neq\zero$ equals \eqref{eq_RG_D1}, \eqref{eq_RG_D2} in the free variables $\mPi$, $\mGamma$. These facts can be checked by properly taking $\Q_x$ and $\D$ in the proofs in \secref{sec_oor} and \appref{app_B} into account.

\subsection{If Regulator Equations $\RG$ Have a Unique Solution}
\label{sec_squ}

\ldots\,given by $(\mPi,\mGamma)$ then it solves \oorref{oor} uniquely. This is a consequence of \lemref{lem_solv_regoa} (see also the part of the proof covering necessity), \thmref{thm_oa} and the fact that \eqref{eq_feasvaroa} is only satisfied for $(\dPi,\dGamma)=(\zero,\zero)$. The latter and \thmref{thm_ua} imply that $(\mPi,\mGamma)$ also solves \oorref{uor} uniquely. We remark that each solution of $\RG$ solves \oorref{uor} if there are more than one. Of course, $(\mPi,\mGamma)$ can also be obtained from \lqtref{lqt} referring to \thmref{thm_lqt}. Hence, in this case, \orref{or} is nothing else but a special case of \lqtref{lqt}.%

\subsection{General Admissible Solutions}
\label{sec_var}
We regard general admissible solutions of \eqref{eq_linsys} that include continuous and piecewise continuously differentiable $\xs(\cdot)$ (with arbitrary $\xs(0)$) and piecewise continuous $\us(\cdot)$ such that $\us(t)$ is bounded on each finite interval. Our optimality definition based on the power difference \eqref{eq_delta_power} can only distinguish differences $\Delta J_T\cd$ that grow at least linearly with $T$. In this more general context, a higher precision is desirable. Thus, we use $\Delta J_T\cd$ itself as a measure as $T\to\infty$.

Let us consider \lqtref{lqt} under the assumption that $(\C,\A)$ is completely observable. For a solution $(\mPi\xol,\mGamma\xol)$ obtained from \eqref{eq_reg1} and \eqref{eq_nec}, it holds $\lim_{T\to\infty}\Delta J_T(\xs,\us,1,\mPi\xol,\mGamma\xol)=\infty$ for any $(\xs,\us)$ such that $\limsup_{t\to\infty}\Vert\xs-\mPi\xol\Vert_2>0$, see \cite[Coroll. 10]{Bernhard2017b}. By \thmref{thm_lqt}, we expect similar results for our candidates with respect to \oorref{oor} and \oorref{uor}, respectively.

 In view of \oorref{oor}, a pair $(\mPi^\oa,\mGamma^\oa)$ given by $\RGoa$ satisfies $\lim_{T\to\infty}\Delta J^u_T(\xs,\us,1,\mPi^\oa\xol,\mGamma^\oa\xol)=\infty$ for any $(\xs,\us)$ such that \eqref{eq_constr} and $\limsup_{t\to\infty}\Vert\xs-\mPi^\oa\xol\Vert_2>0$ hold. This can be similarly proven as \cite[Thm. 9]{Bernhard2017b} if $(\C,\A)$ is completely observable. 
 
 With respect to \oorref{uor} and a pair $(\mPi^\ua,\mGamma^\ua)$ given by $\RGua$, it follows that $\lim_{T\to\infty}\Delta J^y_T(\xs,\us,1,\mPi^\ua\xol,\mGamma^\ua\xol)=\infty$ for any $(\xs,\us)$ such that $\Vert\xs(t)\Vert_2$ and $\Vert\us(t)\Vert_2$ are bounded for every $t\in[0,\infty)$ and $\limsup_{t\to\infty}\Vert\C(\xs-\mPi^\ua\xol)\Vert_2>0$. This can be proven by a contradiction argument using \textit{Barbalat's Lemma} \cite{Khalil2002}. 
 
In this sense, cf. \cite{Artstein1985}, $(\mPi^\oa\xol,\mGamma^\oa\xol)$ and $(\mPi^\ua\xol,\mGamma^\ua\xol)$ \textit{overtake} any feasible $(\xs,\us)$ that differs in the stationary behavior of the state and the output, respectively.
 Hence, our results show strong properties in a general infinite-horizon optimal tracking setup. This underlines that the control structure~\eqref{eq_u} in output regulation theory is not restrictive but rather necessary to obtain such a desirable performance.%

\subsection{If Assumption 2 Does Not Hold}
\label{sec_unstable}

This is true, e.g., for every nonconstant polynomial or unstable reference $\yol(t)$. Of course, it results $P(\cdot)=\lim_{T\to\infty}\nicefrac{J_T(\cdot)}{T}=\infty$ generally. Hence, we seek an \textit{overtaking optimal} stationary solution which overtakes any other stationary solution in the sense of \secref{sec_var}. But, it is easy to construct counter examples to illustrate that, in general, such a solution does not exist for our problems. 

For simplicity, we only give a counter example for \lqtref{lqt} and a polynomial reference. Consider $J_T(\cdot)=\int_0^T(x_\text{s}-\xoli_1)^2+u_\text{s}^2\dop t$ for $\dot x=x-u$ and $\xol^\tp=\ma \nicefrac{t^2}{2} & t & 1 \me$. 
For an admissible control $u_\text{s}=\mGamma\xol=\gamma_1\nicefrac{t^2}{2}+\gamma_2t+\gamma_3$, we obtain $x_\text{s}=\mPi\xol=\gamma_1\nicefrac{t^2}{2}+(\gamma_1+\gamma_2)t+(\gamma_1+\gamma_2+\gamma_3)$ from $\dot x_\text{s}=x_\text{s}-u_\text{s}$ where constant $\gamma_i\in\RZ$ are left to choose. By integration, $J_T(x_\text{s},u_\text{s})$ is a fifth-order polynomial in $T$, where the coefficients are functions of $\gamma_i$. Since $T$ can be arbitrarily large, the best one can do is to choose $\gamma_1$,  $\gamma_2$,  $\gamma_3$ such that the coefficients are minimized stepwise starting from the highest order. This forces $\gamma_1=0.5$ and $\gamma_2=-0.125$. Then the coefficient that depends on $\gamma_3$ and belongs to the highest order, which is $T^2$ in this case, is a linear function of $\gamma_3$. For large $T$, $J_T(x_\text{s},u_\text{s})$ behaves as a linear function of $\gamma_3$ and an overtaking optimal stationary solution cannot exist.

These considerations justify to require \assref{ass_eig} because we cannot expect to find adequate stationary solutions without it. 
\remark Nonetheless, when \assref{ass_eig} is violated as above, then our derived candidates might still be a favorable choice. We could verify this by checking if the candidate $\big(\mPi(\epsilon,\rho)\xol,\mGamma(\epsilon,\rho)\xol\big)$ for \lqtref{lqt} is a so-called \textit{agreeable plan} for $\x_0=\mPi(\epsilon,\rho)\xol_0$ (see \cite{Carlson1987} and \cite{Bernhard2017b} for details) if either $\rho$ is exceedingly large or $\epsilon$ is exceedingly small. This can be expected if $\Re{\lambdaol}=0$ holds for all $\lambdaol\in\sigma(\Aol)$.
 \section{Conclusion}
\label{sec_concl}

Under common assumptions, we derived new design methods for a trajectory tracking control \eqref{eq_u} for linear systems that are square, over-actuated or under-actuated. In this respect, we
contributed two new sets of regulator equations $\RGoa$ and $\RGua$ to the output regulation theory. By solving $\RGoa$, we easily obtain the solution $(\mPi,\mGamma)$ of the classical regulator equations $\RG$ that uses additional actuators most efficiently. If $\RG$ have no solution because exact tracking is infeasible, then $\RGua$ provide a pair $(\mPi,\mGamma)$ that optimally saves tracking error energy. 
Our thorough study of solvability conditions also yielded a significant insight: The classical output regulation problem (\orref{or}) can always be solved by using our new equations $\RGoa$ instead of the classical $\RG$.
Furthermore, we established a link to optimal tracking by showing that both optimal ORP (\oorref{oor} and \oorref{uor}) are in fact special cases of a classical infinite-horizon LQ tracking problem (\lqtref{lqt}). This is useful, e.g., if tracking with (arbitrarily) high precision is sufficient instead of exact tracking.

\appendices
\section{}\label{app_A}
\setcounter{equation}{0}
\renewcommand{\theequation}{A.\arabic{equation}}
\setcounter{lemma}{0}
\renewcommand*{\thelemma}{A.\arabic{lemma}}

The following lemmata are used to prove our main results.
\begin{lemma}\label{lem_Wo}
	For $\G\in\CZ^{n_G\times\nol}$ with $n_G\geq1$, consider the system $(\G,\Aol)$ for which some but not all eigenvalues of $\Aol$ are observable. Due to \assref{ass_eig}, let the exosystem \eqref{eq_exosys} be given by a diagonalization:
	\begin{equation}\label{eq_exo_decomp}
	\ma \dot\xol_1 \\ \dot\xol_2 \me = \ma \Aol_{11} & \zero \\ \zero & \Aol_{22} \me \ma \xol_1 \\ \xol_2 \me
	\end{equation}
	where $\Aol_{11}$ and $\Aol_{22}$ are diagonal such that $\G=\ma \G_1 & \zero \me$ and $(\G_1,\Aol_{11})$ is completely observable. For all $\xol_2(0)$, it holds
	\begin{subnumcases}{\label{eq_power_both}
		\lim_{T\to\infty}\frac{1}{T}\int_0^T \xol(t)^\tp\G^\tp\G\xol(t)\dop t}
		>0\;\text{ if}\;\;\xol_1(0)\neq\zero\textkom\hspace{10mm}\label{eq_power}\\
		=0\;\text{ if}\;\;\xol_1(0)=\zero. \label{eq_power_zero}
	\end{subnumcases}
\end{lemma}
\vspace{2.0mm}
\begin{proof} 
	From the structure of $\Aol$ and $\G$, we observe that 
	$$\frac{1}{T}\int_0^T \xol_1(t)^\tp\G_1^\tp\G_1\xol_1(t) \dop t=\frac{1}{T}\xol_1(0)^\tp\W(T)\xol_1(0)$$ where $\W(T)$ is the observability gramian of the system $(\G_1,\Aol_{11})$. We trivially conclude \eqref{eq_power_zero}. With $(\G_1,\Aol_{11})$ completely observable, it holds $\W(T)\succ\zero$, $\forall T>0$ \cite{Sontag1998}, and the integrand satisfies $\xol_1(t)^\tp\G_1^\tp\G_1\xol_1(t)\not\equiv\zero$ for all $\xol_1(0)\neq\zero$. Taking \assref{ass_eig} into account, the components of $\xol(t)$ are sinusoids for all $\xol_0$. Hence, for every given $\xol_1(0)\neq\zero$, we can rewrite the left side of \eqref{eq_power_both}:
	\begin{equation*}
	\sum_{i}\bigg(\lim_{T\to\infty}\frac{1}{T}\int_0^T\bigg(\sum_j A_{ij}\cos(\omega_{ij}t+\alpha_{ij})\bigg)^2\dop t\bigg)\textkom
	\end{equation*}
	where 
	$A_{ij},\omega_{ij} ,\alpha_{ij}\in\RZ$ and $A_{ij}\neq0$. 
	Then, \eqref{eq_power} follows straightforwardly since it is well known that the signal power of such sums of sinusoids is larger than zero.
\end{proof}

The next lemma is a consequence of standard results:
\begin{lemma}\label{lem_Axb}
	Consider a system of $n\nol\geq1$ equations given by $\mA\,\vx+\vx\mAol=\vb$, where $\mA\in\RZ^{n\times n}$, $\mAol\in\RZ^{\nol\times \nol}$ and $\vb\in\RZ^{n\times \nol}$, and its disturbed version $(\mA+\epsilon\DA) \tvx+\tvx(\mAol+\epsilon\DAol)=\vb$, where $\DA\in\RZ^{n\times n}$, $\DAol\in\RZ^{\nol\times \nol}$ and $\epsilon>0$. Suppose $\vx\in\RZ^{n\times \nol}$ and $\tvx(\epsilon)\colon(0,\infty)\rightarrow\RZ^{n\times\nol}$ are their unique solutions (for every $\epsilon>0$). 
	Then, it holds $\Vert\tvx(\epsilon)-\vx\Vert_\text{F}=\mathcal{O}(\epsilon)$ as $\epsilon\to0$.
\end{lemma}
\begin{proof}
	By using the Kronecker sum $\oplus$ and column-stacking operator $\text{vec}\cd$, both systems can be equivalently rewritten, e.g., $(\mAol^\tp\oplus\mA)\text{vec}(\vx)=\text{vec}(\vb)$. Due to the uniqueness of the solutions, which implies invertibility of $\mAol^\tp\oplus\mA$, we exploit \cite[Fact 9.15.2]{Bernstein2009} 
	to conclude  that if $\epsilon<\nicefrac{1}{\Vert\txtfor{(\mAol^\tp\oplus\mA)^{-1}(\DAol^\tp\oplus\DA)}\Vert_2}$, then
	\begin{equation}\label{eq_bound_Axb}
	\hspace{-3.0mm}\Vert\text{vec}(\txtfor{\tvx(\epsilon)\hspace*{-.75mm}-\hspace*{-.75mm}\vx})\Vert_2\leq\epsilon\,\tfrac{\Vert\txtfor{(\mAol^\tp\oplus\mA)^{-1}}\Vert_2\Vert\txtfor{(\DAol^\tp\oplus\DA)}\Vert_2\Vert\text{vec}(\txtfor{\vx})\Vert_2}{1-\epsilon\txtfor{\Vert(\mAol^\tp\oplus\mA)^{-1}(\DAol^\tp\oplus\DA)}\Vert_2}\hspace{-.5mm}
	\end{equation}
	holds. This implies the big $\mathcal{O}$ notation with $\Vert\cdot\Vert_\text{F}=\Vert\text{vec}\cd\Vert_2$.
\end{proof}
\section{}
\label{app_B}

Here, we present the proofs of \lemref{lem_solv_regoa} and~\ref{lem_nonres_ua}. First, 
let us define the system matrices of the system \eqref{eq_sys} by $\Ross{s}\coloneqq\Ros{s,\C,\A,\B}$ and of the costate system associated with \eqref{eq_oa1} for an artificial output $-\B^\tp\phiv$ by $\Rosc{s}\coloneqq\Ros{-s,-\B^\tp,\A^\tp,-\C^\tp}$. By merging both, the \textit{Hamiltonian} system associated with $\RGoa$ is obtained. After simple manipulations, its system matrix can be given by%
\begin{equation*}
	\Rosh{s}\coloneqq\ma \zero & -s\I-\A^\tp & \C^\tp \\
							s\I-\A & \B \R^{-1} \B^\tp & \zero\\
							\C & \zero & \zero \me.
\end{equation*}

\begin{proof}[Proof of \lemref{lem_solv_regoa}]
From a triple $(\mPi^\oa,\mPic^\oa,\mGammac^\oa)$ solving $\RGoa$ we construct the pair $(\mPi^\oa,-\R^{-1}\B^\tp\mPic^\oa)$ that solves \eqref{eq_reg} due to \eqref{eq_oa2}, \eqref{eq_oa3} which verifies necessity.

Before we show sufficiency, let us reformulate the sets of equations in question.  Due to \assref{ass_eig}, assume that $\Aol$ is diagonal without loss of generality. By \cite[Thm. 1.9]{Huang2004}, it is well known that solving $\RG$  is equivalent to solving $\vec{\beta}_i=\Ross{\lambdaol_i}\zsi$ for $\zsi\in\CZ^{n+m}$ for each $\lambdaol_i \in \sigma\left(\Aol\right)$, $i=1,\ldots,\nol$. Each $\vec{\beta}_i\in\CZ^{n+p}$ depends on $\Ed\vec{v}_i$, $\Col\vec{v}_i$ and $\Dd\vec{v}_i$ with the eigenvector $\vec{v}_i\in\CZ^{\nol}$ of $\Aol$ associated with the eigenvalue $\lambdaol_i$. Accordingly, solving $\RGoa$ is equivalent to solving $\big[ \zero \;\; \vec{\beta}_i^\tp \,\big]^\tp=\Rosh{\lambdaol_i}\zhi$  for $\zhi\in\CZ^{2n+p}$ and $\forall i=1,\ldots,\nol$.

In order to prove the sufficiency, it clearly suffices to show that for all $\lambdaol\in\sigma\left(\Aol\right)$ the equation $\big[\zero\;\; \vec{\beta}^\tp \,\big]^\tp =\Rosh{\lambdaol}\zh$ has a solution if $\vec{\beta} =\Ross{\lambdaol}\zs$ has a solution for some $\betav\in\CZ^{n+p}$. By \cite[Fact 2.10.6]{Bernstein2009}, the latter is true if and only if $\vec{\nu}_\s^\hp\betav=\zero$ for all $\vec{\nu}_\s\in\lnulls=\{\vec{\nu}\in\CZ^{n+p}\,\vert\,\zero=\vec{\nu}^\hp\Ross{\lambdaol}\}$. 
Now, let us consider the set
\begin{multline*}
\set{S}_{\lambdaol}=\left\{ \widehat{\vec{\nu}}\coloneqq\ma \vec{\nu}_\ct \\ \vec{\nu}_\s \me \in \CZ^{2n+p} \;\bigg\vert\;
\vec{\nu}_\s\in\lnulls\right\}.
\end{multline*}
Suppose $\lnullh\subset\set{S}_{\lambdaol}$ holds. By this assumption, we find 
$$\widehat{\vec{\nu}}^\hp\ma \zero \\ \vec{\beta} \me=\ma \vec{\nu}_\ct \\ \vec{\nu}_\s \me^\hp \ma \zero \\ \vec{\beta} \me=\vec{\nu}_\s^\hp\betav=\zero\textkom \;\;\forall \widehat{\vec{\nu}} \in \set{S}_{\lambdaol}\textkom$$
 which implies that $\ma \zero & \vec{\beta}^\tp \me^\tp =\Rosh{\lambdaol}\zh$ has a solution if a solution of $\vec{\beta} =\Ross{\lambdaol}\zs$ exists.

To complete the proof, it only remains to show that $\forall \lambdaol\in\sigma(\Aol)$ $\lnullh\subset\set{S}_{\lambdaol}$ holds indeed. Let $\widehat{\vec{\nu}}=\big[ \vec{\nu}_\ct^\hp \;\; \vec{\nu}_{\s1}^\hp \;\; \vec{\nu}_{\s2}^\hp \big]$, this is true if $\forall \lambdaol\in\sigma(\Aol)$: $\widehat{\vec{\nu}}^\hp\Rosh{\lambdaol}=\zero$ implies $\big[ \vec{\nu}_{\s1}^\hp \;\; \vec{\nu}_{\s2}^\hp \big]\Ross{\lambdaol}=\zero$. By introducing $\widetilde{\vec{\nu}}_{\s1}^\hp=-\vec{\nu}_{\s1}^\hp\B\R^{-1}$ with $\widetilde{\vec{\nu}}_{\s1}\in\CZ^m$, we may rewrite $\widehat{\vec{\nu}}^\hp\Rosh{\lambdaol}=\zero$ equivalently:%
\begin{subequations}\label{eq_ham_decomp}\begin{align}
	\ma \vec{\nu}_\ct^\hp & \widetilde{\vec{\nu}}_{\s1}^\hp \me \Rosc{\lambdaol} &= \zero\textkom \label{eq_ham_costate}\\
	\ma \vec{\nu}_{\s1}^\hp & \vec{\nu}_{\s2}^\hp \me \Ross{\lambdaol} &= \ma \zero & \widetilde{\vec{\nu}}_{\s1}^\hp\R \me \label{eq_ham_constr},
	\end{align}\end{subequations}
which corresponds to a decomposition of the \textit{Hamiltonian} system into a series of costate system and original system.

Since $\Rosc{\lambdaol}=\Ross{\lambdaol}^\hp$ $\forall \lambdaol\in\sigma(\Aol)$, due to $\lambdaol\in j\RZ$ based on \assref{ass_eig}, we find that $\big[\vec{\nu}_\ct^\hp\;\;\widetilde{\vec{\nu}}_{\s1}^\hp\big]^{\hp} \in \nulls\big(\Ross{\lambdaol}\big)$ by taking the conjugate transpose of \eqref{eq_ham_costate}. Thus, we  obtain from \eqref{eq_ham_constr}:
$$\ma \zero & \widetilde{\vec{\nu}}_{\s1}^\hp\R \me\ma \vec{\nu}_\ct \\ \widetilde{\vec{\nu}}_{\s1} \me=\widetilde{\vec{\nu}}_{\s1}^\hp\R\widetilde{\vec{\nu}}_{\s1}=\zero.$$
Hence, \eqref{eq_ham_decomp} necessarily implies $\widetilde{\vec{\nu}}_{\s1}=\zero$ because of $\R\succ\zero$ and $\ma \vec{\nu}_{\s1}^\hp & \vec{\nu}_{\s2}^\hp \me \Ross{\lambdaol} = \zero$ follows from \eqref{eq_ham_constr} as desired.
\end{proof}

Before we proceed to prove \lemref{lem_nonres_ua}, we define the system matrix of the \textit{Hamiltonian system} associated with $\RGua$ by
\begin{equation*}
\Roshu{s}\coloneqq\ma \zero & s\I-\A & \B \\
-s\I-\A^\tp & -\C^\tp \Qy \C & \zero\\
-\B^\tp & \zero & \zero \me.
\end{equation*}

\begin{proof}[Proof of \lemref{lem_nonres_ua}]
	By using similar techniques as in the proof of \lemref{lem_solv_regoa}, showing that $\RGua$ always has a \textit{unique} solution is equivalent to showing that a \textit{unique} solution $\zh\in\CZ^{2n+m}$ to $\big[ \vec{\beta}_\s^\tp \;\; \vec{\beta}_\ct^\tp\C \;\; \zero \big]^\tp=\Roshu{\lambdaol}\zh$ exists for all $\lambdaol \in \sigma\left(\Aol\right)$, all $\vec{\beta}_\s\in\CZ^{n}$ and all $\vec{\beta}_\ct\in\CZ^{p}$. This is true if and only if $\Roshu{\lambdaol}$ has full rank $\forall\lambdaol\in\sigma({\Aol})$, i.e., $\Roshu{\lambdaol}\widehat{\vec{\nu}}=\zero$ with $\widehat{\vec{\nu}}\in\CZ^{2n+m}$ admits only a trivial solution $\widehat{\vec{\nu}}=\zero$. 
	Let $\widehat{\vec{\nu}}^\tp=\ma \vec{\nu}_\s^\tp & \vec{\nu}_{\ct1}^\tp & \vec{\nu}_{\ct2}^\tp \me$, we rewrite $\Roshu{\lambdaol}\widehat{\vec{\nu}}=\zero$ by a similar decomposition as in the proof above:
	\begin{subequations}\label{eq_ham_decomp_ua}\begin{align}
		\Ross{\lambdaol} \ma \vec{\nu}_{\ct1} \\ \vec{\nu}_{\ct2} \me &= \ma \zero \\ -\Qy^{-1}\widetilde{\vec{\nu}}_{\ct1} \me\textkom \label{eq_ham_costate_ua}\\
		 \Rosc{\lambdaol}\ma \vec{\nu}_{\s} \\ \widetilde{\vec{\nu}}_{\ct1} \me &= \zero \label{eq_ham_constr_ua}
		\end{align}\end{subequations}
	where $\widetilde{\vec{\nu}}_{\ct1}\in\CZ^p$ such that $\widetilde{\vec{\nu}}_{\ct1}=-\Qy\C\vec{\nu}_{\ct1}$.
	From \eqref{eq_ham_constr_ua} and the fact that $\Rosc{\lambdaol}^\hp=\Ross{\lambdaol}$ $\forall\lambdaol\in\sigma(\Aol)$ due to \assref{ass_eig}, we find that $\ma \vec{\nu}_{\s}^\hp & \widetilde{\vec{\nu}}_{\ct1}^\hp \me\Ross{\lambdaol} = \zero$. Then, \eqref{eq_ham_costate_ua} has a solution only if
	\begin{equation*}
		\ma \vec{\nu}_{\s}^\hp & \widetilde{\vec{\nu}}_{\ct1}^\hp \me\ma \zero \\ -\Qy^{-1}\widetilde{\vec{\nu}}_{\ct1} \me=-\widetilde{\vec{\nu}}_{\ct1}^\hp\Qy^{-1}\widetilde{\vec{\nu}}_{\ct1}=\zero
	\end{equation*}	
	holds. Clearly, this requires $\widetilde{\vec{\nu}}_{\ct1}=\zero$ and, consequently,
	$$\Rosc{\lambdaol}\ma \vec{\nu}_{\s} \\ \zero \me =\zero\;\text{ and }\;\Ross{\lambdaol} \ma \vec{\nu}_{\ct1} \\ \vec{\nu}_{\ct2} \me=\zero.$$ 
	Since $-\lambdaol^\hp=\lambdaol\in j\RZ$ holds for all $\lambdaol\in\sigma(\Aol)$ due to \assref{ass_eig}, the first equation gives $\vec{\nu}_{\s}^\hp\ma \lambdaol\I-\A & -\B \me=\zero$. This forces $\vec{\nu}_{\s}=\zero$ by \assref{ass_stab}. The second equation admits only the trivial solution $\vec{\nu}_{\ct1}=\zero$ and $\vec{\nu}_{\ct2}=\zero$ $\forall\lambdaol \in \sigma\left(\Aol\right)$ if and only if \assref{ass_nonres_ua} holds.
\end{proof}

\ifCLASSOPTIONcaptionsoff
\newpage
\fi

\bibliographystyle{IEEEtranS} 
\bibliography{literature}

\end{document}